\documentclass[conference]{IEEEtran}
\IEEEoverridecommandlockouts
\usepackage{cite}

\ifCLASSINFOpdf
  \usepackage[pdftex]{graphicx}
  \graphicspath{{../pdf/}{../jpeg/}}
  \DeclareGraphicsExtensions{.pdf,.jpeg,.png}
\else
\fi
\usepackage{algorithm}
\usepackage{algpseudocode}
\usepackage{amsmath}
\usepackage{amssymb}
\usepackage{amsthm}
\usepackage{hyperref}
\usepackage{paralist}
\usepackage{color}

\newcommand{\vx}{{\vec{x}}}
\newcommand{\vy}{{\vec{y}}}
\newcommand{\vz}{{\vec{z}}}
\newcommand{\vw}{\vec{w}}
\newcommand{\hx}{{\hat{x}}}
\newcommand{\hy}{{\hat{y}}}
\newcommand{\hz}{{\hat{z}}}
\newcommand{\hw}{\hat{w}}
\newtheorem{definition}{Definition}

\newtheorem{lemma}{Lemma}
\newtheorem{theorem}{Theorem}
\newtheorem{corollary}{Corollary}
\newtheorem{example}{Example}
\newtheorem{problem}{Problem}

\begin{document}
%
\title{Functional Synthesis via Input-Output Separation
\thanks{Work supported in part by NSF grants CCF-1319459 and IIS-1527668, by NSF Expeditions in Computing project "ExCAPE: Expeditions in Computer Augmented Program Engineering", by a grant from MHRD, Govt of India, under the IMPRINT-1 scheme, and by the Brazilian agency CNPq through the Ci\^{e}ncia Sem Fronteiras program.}}


\author{\IEEEauthorblockN{Supratik Chakraborty}
\IEEEauthorblockA{
\textit{IIT Bombay}\\
Mumbai, India \\
supratik@cse.iitb.ac.in}
\and
\IEEEauthorblockN{Dror Fried}
\IEEEauthorblockA{
\textit{Rice University}\\
Houston, USA \\
dror.fried@rice.edu}
\and
\IEEEauthorblockN{Lucas M. Tabajara}
\IEEEauthorblockA{
\textit{Rice University}\\
Houston, USA \\
lucasmt@rice.edu}
\and
\IEEEauthorblockN{Moshe Y. Vardi}
\IEEEauthorblockA{
\textit{Rice University}\\
Houston, USA \\
vardi@cs.rice.edu}
}


%


\maketitle

\begin{abstract}
Boolean functional synthesis is the process of constructing a Boolean function from a Boolean specification that relates input and output variables. Despite significant recent developments in synthesis algorithms, Boolean functional synthesis remains a challenging problem even when state-of-the-art methods are used for decomposing the specification.
In this work we bring a fresh decomposition approach, orthogonal to existing methods, that explores the decomposition of the specification into separate input and output components. 
We make  use of an input-output decomposition of  a given specification described as a  CNF formula,  by alternatingly analyzing the separate input and output components. We exploit well-defined properties of these components to ultimately synthesize a solution for the entire specification.
We first provide a theoretical result that, for input components with specific structures, synthesis for CNF formulas via this framework can be performed more efficiently  than in the general case.
We then show by experimental evaluations that our algorithm performs well also in practice on instances which are challenging for existing state-of-the-art tools, serving as a good complement to modern synthesis techniques.
\end{abstract}


%
\IEEEpeerreviewmaketitle

\section{Introduction}\label{sec:intro}
	
Boolean functional synthesis is the problem of constructing a Boolean function from a Boolean specification that describes a relation between input and output variables~\cite{JSCTA15,FTV16,ACJS17,TV17}. This problem has been explored in a number of settings including circuit design~\cite{KS00}, QBF solving~\cite{RS2016}, and reactive synthesis~\cite{ZTLPV17}, and several tools have been developed for its solution. Nevertheless, scalability of Boolean functional synthesis methods remains a concern as the number of variables and size of the formula grows. This is not surprising since Boolean functional synthesis is in fact \textsc{co-NP\textsuperscript{NP}}-hard.
	
A standard practice for handling the problem of scalability is based on decomposing the given formula into smaller sub-specifications and synthesizing each component separately~\cite{JSCTA15,ACJS17,TV17}. The most common form of such decomposition, called \emph{factorization}, is when the formula is represented as a conjunction of constraints, in which  each conjunct can be seen as a sub-specification~\cite{JSCTA15,TV17}. The main challenge in this approach is that most factors cannot be synthesized entirely separately due to the dependencies created by shared input and output variables. The ways to meet this challenge are usually to either merge factors that share variables~\cite{TV17} or perform additional computations in order to combine the functions synthesized for different factors~\cite{JSCTA15}. All these result in additional work that must be performed during the synthesis. 

In this work, we propose an alternative decomposition framework,  which follows naturally from the fact that variables in the specification are separated into input and output variables. This idea was originally inspired by~\cite{FLOV18}, which explores the notion of \emph{sequential relational decomposition}, in which a relation is decomposed into two by introducing an intermediate domain. Differently from factorization, this form of decomposition allows the two components to be synthesized completely independently. That work, however, shows that decomposition is hard in general, and if the relation is given as a Boolean circuit, decomposition is NEXPTIME-complete. Furthermore, there is no guarantee that synthesizing the two components independently would be easier than synthesizing the original specification, since the synthesis of one component might ignore useful information given by the other component.

We instead suggest a more relaxed notion of decomposition for specifications described as CNF formulas, in  which every clause is split into an input and an output clause and the independent analyses of the input/output components ``cooperate'' to synthesize a function for the entire specification. Based on this concept, we describe a novel synthesis algorithm for CNF formulas called the ``Back-and-Forth'' algorithm, where rather than synthesizing the input and output components entirely independently we share information back and forth between the two components to guide the synthesis. More specifically,  our algorithm alternates between SAT calls that follow the input-component structure analysis and MaxSAT calls that follow the output-component structure analysis. Thus, this approach builds on recent progress with SAT and MaxSAT solving \cite{SLM2009,LiMan09}. A notable consequence of our method  is that, as the number of SAT calls is dependent on the structure of the input component, for specifications with some well-defined input structure we can perform synthesis in \textsc{P\textsuperscript{NP}}, compared to the generally mentioned \textsc{co-NP\textsuperscript{NP}}-hardness. An additional advantage of our algorithm is that it constructs the synthesized function as a \emph{decision list}~\cite{Rivest87}. Compared to other data structures for representing Boolean functions, such as ROBDDs or AIGs, decision lists have significant benefits in term of explainability, allowing domain specialists to validate and analyze their behavior (see discussion in Section \ref{sec:discussion} for more details).
  
We experimentally evaluate the ``Back-and-Forth'' algorithm on a suite of standard synthesis benchmarks, comparing its performance with that of state-of-the-art synthesis tools. Although these tools perform very well on many families of benchmarks, our results show that the ``Back-and-Forth'' algorithm is able to handle classes of benchmarks that these tools are unable to synthesize, indicating that it belongs in a portfolio of synthesis algorithms.
  

  
\section{Related Work}\label{sec:related}
	
    
Constructing explicit representations of implicitly specified functions is a fundamental problem of interest to both theoreticians and practitioners. In the contexts of Boolean functional synthesis and certified QBF solving, such functions are also called \emph{Skolem functions}~\cite{BWJ2014,JSCTA15,HSB2014}. Boole~\cite{Boole1847} and Lowenheim~\cite{Lowenheim1910} studied variants of this problem when computing most general unifiers in resolution-based proofs.  Unfortunately, their algorithms, though elegant in theory, do not scale well in practice~\cite{MOP1998}. The close relation between Skolem functions and proof objects in specialized QBF proof systems has been explored in~\cite{BWJ2014,HSB2014}. One of the earliest applications of Boolean functional synthesis has been logic synthesis - see~\cite{Tabajara18} for a survey. More recently, Boolean functional synthesis has found applications in diverse areas such as temporal strategy synthesis~\cite{AMN2005,SYNTCOMP15,ZTLPV17}, certified QBF solving ~\cite{JB2011,RT2015,BJ2012,NPLSB2012}, automated program synthesis~\cite{SGF2013,SoLe2013}, circuit repair and debugging~\cite{JMF2012}, and the like.  This has resulted in a new generation of Boolean functional synthesis tools, cf. ~\cite{HSB2014,JSCTA15,ACJS17,ACGKS2018,FTV16,TV17,RT2015,RS2016}, that are able to synthesize functions from significantly larger relational specifications than what was possible a decade back.

Recent tools for Boolean functional synthesis can be broadly categorized based on the techniques employed by them. Given a specification $F(\vx, \vy)$, where $\vx$ denotes inputs and $\vy$ denotes outputs, the work of~\cite{HSB2014} extracts Skolem functions for $\vy$ in terms of $\vx$ from a proof of validity of $\forall \vx.\exists\vy.F(\vx,\vy)$ expressed in a specific format. The efficiency of this technique crucially depends on the existence and size of a proof in the required format.  \emph{Incremental  determinization}~\cite{RS2016} is a highly effective synthesis technique that accepts as input a CNF representation of a specification and
builds on several successful heuristics used in modern conflict-driven clause-learning (CDCL) SAT solvers~\cite{SLM2009}.

In~\cite{FTV16}, the composition-based synthesis approach of~\cite{J2009} is adapted and new heuristics are proposed for synthesizing Skolem functions from an ROBDD representation of the specification.  The technique has been further improved in~\cite{TV17} to work with factored specifications represented as implicitly conjoined ROBDDs. CEGAR-based techniques that use modern SAT solvers as black boxes~\cite{JSCTA15,ACJS17,ACGKS2018} have recently been shown to scale well on several classes of large benchmarks.  The idea behind these techniques is to start with an efficiently computable initial estimate of Skolem functions, and use a SAT solver to test if the estimates are correct. A satisfying assignment returned by the solver provides a counterexample to the correctness of the function estimates, and can be used to iteratively refine the estimates.  In~\cite{ACGKS2018}, it is shown that transforming the representation of the specification to a special negation normal form allows one to efficiently synthesize Skolem functions.

Both ROBDD and CEGAR-based approaches make use of decomposition techniques to improve performance, the most common of which is \emph{factorization}~\cite{JSCTA15,TV17}. In this method, every conjunct of a conjunctive specification is considered individually. The main drawback in this approach is that the dependencies between conjuncts limit how much each of them can be analyzed independently of the others, requiring either partially combining components, as in~\cite{TV17}, or going through a process of refinement of the results~\cite{JSCTA15}. This issue motivates the search for alternative notions of decomposition for synthesis problems.
Our approach is loosely inspired by the idea of \emph{sequential relational decomposition} explored in depth in~\cite{FLOV18}. A more direct application of this idea to synthesis might still be possible, but requires further exploration.
In addition to the above techniques, templates or sketches have been used to synthesize functions when information about the possible functional forms is available a priori~\cite{SGF2013,SRBE2005}. 

As is clear from above, several orthogonal techniques have been found to be useful for the Boolean functional synthesis problem. In fact, there remain difficult corners, where the specification is stated simply, and yet finding Skolem functions that satisfy the specification has turned out to be hard for all state-of-the-art tools.  Our goal in this paper is to present a new technique and algorithm for this problem, that does not necessarily outperform existing techniques on all benchmarks, but certainly outperforms them on instances in some of these difficult corners.  We envisage our technique being added to the existing repertoire of techniques in a portfolio Skolem-function synthesizer, to expand the range of problems that can be solved.

\section{Preliminaries}\label{sec:Prem}

\subsection{Boolean Functional Synthesis}

A specification for the Boolean functional synthesis problem is a (quantifier-free) Boolean formula $F(\vx, \vy)$ over \emph{input variables} $\vx = (x_1, \ldots, x_m)$ and \emph{output variables} $\vy = (y_1, \ldots, y_n)$. Note that $F$ can be interpreted as a relation $F \subseteq X \times Y$, where $X$ is the set of all assignments $\hx$ to $\vx$ and $Y$ is the set of all assignments $\hy$ to $\vy$. With that in mind, we denote by $Dom(F) = \{\hx \mid \exists \hy . (F(\hx, \hy) = 1)\}$ and $Img(F) = \{\hy \mid \exists \hx . (F(\hx, \hy) = 1)\}$ the domain and image of the relation represented by $F$. We also use $Img_{\hx}(F) = \{\hy \mid F(\hx, \hy) = 1\}$ to denote the image of a specific element $\hx \in X$. If $Dom(F) = X$, then we say that $F$ is \emph{realizable}.

Two Boolean formulas $F(\vw)$ and $F'(\vw)$ are said to be \emph{logically equivalent}, denoted by $F \equiv F'$, if they have the same solution space; that is, for every assignment $\hw$ to $\vw$, $F(\hw) = 1$ iff $F'(\hw) = 1$. Unless stated otherwise, all Boolean formulas mentioned in this work are quantifier free.

We say that a partial function $g : X \to Y$ \emph{implements} a relation $F \subseteq X \times Y$ if for every $\hx \in Dom(F)$ we have that $(\hx, g(\hx)) \in F$. Such a $g$ is also called a \emph{Skolem function} of $F$. Note that if $F$ is realizable, then $g$ is a total function.
Finally, we define the \emph{Boolean-synthesis problem} as follows:

\begin{problem}
Given a specification $F(\vx, \vy)$, construct a partial function $g$ that implements $F$.
\end{problem}

For more information on Boolean synthesis, see~\cite{JSCTA15,FTV16}.

\subsection{Decision lists}

   Our choice of representation of Skolem functions in this work is inspired by the idea that we can represent an arbitrary Boolean function $f$ by a \emph{decision list}~\cite{Rivest87}. A decision list is an expression of the form \texttt{if $f_1(\vx)$ then $\hy_1$ else if $f_2(\vx)$ then $\hy_2$ else $\ldots$ else $\hy_k$}, where each $f_i$ is a formula in terms of the input variables $\vx$ and each $\hy_i$ is an assignment to the output variables $\vy$. The length $k$ of the list corresponds to the number of decisions. Clearly, for a specification $F(\vx, \vy)$ with $m$ input variables we can always synthesize as an implementation a decision list of length $2^m$, where for every possible assignment of $\vx$ we choose an assignment of $\vy$ that satisfies the specification. Many specifications, however,  can be implemented by significantly smaller decision lists, by taking advantage of the fact that multiple inputs can be mapped to the same output. Our analysis identifies and exploits these cases.
   
   Despite being a natural representation, decision lists might not be appropriate for a physical implementation of the synthesized function as a circuit. In this case, it might make sense to collect the decisions into a more compact representation, such as an ROBDD.

\subsection{Conjunctive Normal Form}

A Boolean formula $F(\vw)$ is in \emph{conjunctive normal form} (CNF) if $F$ is a conjunction of clauses $C_1 \land \ldots \land C_k$, where every clause $C_i$ is a disjunction of literals (a variable or its negation). A subset $S$ of the clauses of a CNF formula $F$ is \emph{satisfiable} if there exists an assignment $\hw$ to the variables $\vw$ in $F$ such that $C_i(\hw) = 1$ for every clause $C_i \in S$. Similarly, a subset $S$ of the clauses of $F$ is \emph{all-falsifiable} if there exists an assignment $\hw$ such that  $C_i(\hw) = 0$ for every clause $C_i \in S$. A subset $S$ of clauses is a \emph{maximal satisfiable subset} (MSS) if $S$ is satisfiable and every superset $S' \supset S$ is unsatisfiable. Similarly, $S$ is a \emph{maximal falsifiable subset} (MFS) if $S$ is all-falsifiable and every superset $S' \supset S$ is not all-falsifiable. For more information on MSS and MFS, refer to~\cite{IMPM13}.
\newcommand{\Co}{{Co}}

\section{Synthesis via Input-Output Separation}\label{sec:amicable}
	
    In this section, we present a novel algorithm for Boolean functional synthesis from CNF specifications. Our approach is based on a separation of every clause into an input part and an output part. First, we describe how a decision list implementing the specification can be constructed by enumerating MFSs of the input clauses, or similarly by enumerating MSSs of the output clauses. Then, we show how we can benefit from alternating between the two: the MFSs can be used to avoid useless MSSs, while the MSSs can be used to cover multiple MFSs at the same time without enumerating all of them.

    
    Given a CNF formula $F(\vec{x}, \vec{y})$, assume $F(\vec{x}, \vec{y}) = \bigwedge^k_{i = 1} C_i$, where $C_1, \ldots, C_k$ are clauses over $\vec{x}$ and $\vec{y}$.
    Let $C_i|_{\vec{x}}$ denote the $x$-part of clause $C_i$, that is, the disjunction of all $x$ literals in $C_i$.  Similarly, let $C_i|_{\vec{y}}$ be the $y$-part of clause $C_i$, the disjunction of all $y$ literals in $C_i$. We call $S_\vx = \{C_i|_{\vec{x}} \mid C_i \text{ is a clause in $F$}\}$ and $S_\vy = \{C_i|_{\vec{y}} \mid C_i \text{ is a clause in $F$}\}$ the set of input and output clauses of the specification, respectively. 
    
In the following sections, we describe how to perform separate analyses of the input component $S_\vx$ and the output component $S_\vy$, and then how to combine these analyses into a single synthesis algorithm that alternates between the two components.
    
    \subsection{Analysis of the Input Component}\label{sec:input_analysis}

    In this subsection we assume that the specification $F$ is realizable. First, consider a single assignment $\hx$ to the input variables $\vx$. Let $Fals(\hx) = \{ C_i|_\vx \in S_\vx \mid C_i|_\vx(\hx) = 0 \}$ be the subset of input clauses that $\hx$ falsifies. For a set $S'_\vx \subseteq S_\vx$ of input clauses, let $\Co(S'_\vx) = \{ C_i|_\vy \in S_\vy \mid C_i|_\vx \in S'_\vx \}$ be the corresponding set of output clauses and let $MustSat(\hx) = \Co(Fals(\hx))$. Note that $C_i \equiv (C_i|_\vx \lor C_i|_\vy) \equiv (\neg C_i|_\vx \rightarrow C_i|_\vy)$ for every clause $C_i$. Therefore  $MustSat(\hx)$ is the subset of output clauses that must be satisfied in order to satisfy $F$ when $\hx$ is the input assignment.
    
    A key observation is that for two different input assignments $\hx$ and $\hx'$, if $Fals(\hx') \subseteq Fals(\hx)$, then $MustSat(\hx') \subseteq MustSat(\hx)$, and therefore every output assignment $\hy$ that satisfies the specification for $\hx$ also satisfies the specification for $\hx'$. Hence, it is enough to consider only assignments for $\vx$ that falsify a maximal number of input clauses. This leads to the following lemma:
    
	\begin{lemma} \label{lemma:mfs}
		Let $M_\vx$ be an MFS of $S_\vx$, and $\hy$ be an assignment that satisfies $\Co(M_\vx)$. Then: (1) For every assignment $\hx$ such that $Fals(\hx) \subseteq M_\vx$,  the assignment $(\hx, \hy)$ satisfies $F(\vx,\vy)$; and 
			(2) There is no assignment $\hx$ such that $Fals(\hx) \supset M_\vx$.
	\end{lemma}
	
	\begin{proof}
(1) For every clause $C_i|_\vx \in Fals(\hx)$, since $C_i|_\vx \in M_\vx$, we have that $C_i|_\vy$ is in $\Co(M_\vx)$  and therefore is satisfied by $\hy$. Therefore, every clause $C_i$ in $F(\vx, \vy)$ that is not satisfied by $\hx$ is satisfied by $\hy$.
Note that (2) follows from $M_\vx$ being maximal.
\end{proof}
    
    From Lemma~\ref{lemma:mfs} and our assumption that $F(\vx, \vy)$ is realizable, we can conclude the following.
    
    \begin{corollary}
    $F$ can be implemented by a decision list of length equal to the number of MFS of $S_\vx$, where each $f_i$ in the decision list is of size linear in the size of the specification. 
    \end{corollary}
    
    \begin{proof}
    Construct $f_i(\vx)$ by taking the conjunction of all input clauses $C|_\vx$ not contained in the $i$-th MFS $M_i$. Then, $f_i(\vx)$ is satisfied exactly by those assignments $\hx$ such that $Fals(\hx)$ is a subset of $M_i$.
    Then, set the corresponding output assignment $\hy_i$ to an arbitrary satisfying assignment of $\Co(M_i)$.
   \end{proof} 
   
   \begin{example} \label{ex:mfs}
Let $F(x_1, x_2, y_1, y_2) = (x_1 \lor \neg x_2 \lor y_1) \land (x_1 \lor x_2 \lor \neg y_1) \land (x_2 \lor y_1 \lor \neg y_2) \land (\neg x_1 \lor x_2 \lor y_2)$. We first construct input clauses $S_\vx = \{(x_1 \lor \neg x_2), (x_1 \lor x_2), (x_2), (\neg x_1 \lor x_2)\}$ and output clauses $S_\vy = \{(y_1), (\neg y_1), (y_1 \lor \neg y_2), (y_2)\}$. $S_\vx$ has three MFS: $\{(x_1 \lor \neg x_2)\}$, $\{(x_1 \lor x_2), (x_2)\}$ and $\{(x_2), (\neg x_1 \lor x_2)\}$. From these MFS we can construct a decision list implementing $F$ in the way described above. Note that this decision list necessarily covers every possible input assignment:
\begin{align*}
&\texttt{if } (x_1 \lor x_2) \land (x_2) \land (\neg x_1 \lor x_2) \texttt{ then } (y_1 := 1; y_2 := 0) \\
&\texttt{else if } (x_1 \lor \neg x_2) \land (\neg x_1 \lor x_2) \texttt{ then }(y_1 := 0; y_2 := 0) \\
&\texttt{else if } (x_1 \lor \neg x_2) \land (x_1 \lor x_2) \texttt{ then }(y_1 := 1; y_2 := 1)
\end{align*}
\end{example}
    
    
    Note that we require $F(\vx, \vy)$ to be realizable because otherwise we cannot guarantee that $\Co(M_\vx)$ will be satisfiable for every MFS $M_\vx$ of the input clauses. If $\Co(M_\vx)$ is unsatisfiable, however, it is not enough to simply remove the corresponding $f_i(\vx)$ from the decision list, because there might be a subset $M'_\vx \subset M_\vx$ for which $Co(M'_\vx)$ is satisfiable.
    
    This is the first time to our knowledge that MFS are used for synthesis purposes. An advantage of enumerating MFS is that finding an MFS can be easily done, in a precise sense discussed below. One way to do this is through the \emph{conflict graph} of the set of input clauses~\cite{GS17}. Given a set of clauses $S$, the conflict graph of $S$ is the graph where every vertex corresponds to a clause in $S$, and there is an edge between two vertices iff the corresponding clauses have a complementary pair of literals between them (that is, the same variable appears in positive form in one clause and in negative form in the other). The complement of the conflict graph is called a \emph{consensus graph}~\cite{GS17}.
    
    
   Since two clauses can be falsified at the same time iff there is no edge between them in the conflict graph, or alternatively there is an edge between them in the consensus graph, there is a one-to-one correspondence between MFS of the set of clauses, maximal independent sets (MIS) in the conflict graph, and maximal cliques in the consensus graph. Therefore, we can enumerate the MFS in a set of clauses by either enumerating MIS in the conflict graph or maximal cliques in the consensus graph. The benefit of this reduction is that maximal cliques display a so called \emph{polynomial-time listability}, meaning that finding a maximal clique can be performed in polynomial time, and therefore enumeration takes polynomial time in the number of maximal cliques~\cite{IMPM13}.

This relation between the set of  MFS and maximal cliques implies that the size of the smallest decision list that implements a given specification is upper bounded by the number of maximal cliques in the consensus graph of the input clauses.
Therefore we have the following result.

\begin{theorem}\label{thm:compBreech1}
Synthesis can be performed in \textsc{P\textsuperscript{NP}} for specifications for which the consensus graph of $S_\vx$ has a polynomial number of maximal cliques (such as planar or chordal graphs).
\end{theorem}

\begin{proof}
Given a specification $F$, construct the consensus graph of the input component, enumerate the maximal cliques and for each one use a SAT solver to obtain a  corresponding satisfying assignment for the output clauses. Since the number of maximal cliques is polynomial, only a polynomial number of SAT calls is required.
  \end{proof} 
  
Theorem~\ref{thm:compBreech1} demonstrates an improvement relative to the general \textsc{co-NP\textsuperscript{NP}}-hardness of synthesis. Moreover, constructing the consensus graph of the input component is easy, as is testing for certain graph properties, such as planarity, that ensure a small number of maximal cliques. Therefore, Theorem~\ref{thm:compBreech1} provides an elegant method of deciding whether synthesis can be performed efficiently in practice before even beginning the synthesis process.
   
To summarize this section, the analysis of the input component provides two insights. First, a decision list implementing the specification can be constructed from the list of MFS of the input clauses. Second, analyzing the graph structure of the input component allows us to identify classes of specifications for which synthesis can be performed more efficiently. Note that this analysis, however, does not take into account the properties of the output component, and as such the  decision list produced by ignoring the output component may be longer than necessary. With that in mind, the next section presents a complementary analysis of the output component that can help to produce a smaller decision list.
	
	\subsection{Analysis of the Output Component} \label{sec:output_analysis}
	
	For the analysis of the output component, consider the set $MustSat(\hx)$, defined in the previous subsection, of output clauses that must be satisfied when $\hx$ is the input assignment. Then for every two input assignments $\hx$ and $\hx'$, if $MustSat(\hx') \subseteq MustSat(\hx)$,  every output assignment $\hy$ that satisfies the specification for $\hx$ also satisfies the specification for $\hx'$. Therefore, it is enough when constructing the decision list to consider only those satisfiable subsets of $S_\vy$ that are of maximal size. Similarly to Lemma~\ref{lemma:mfs} in the previous section, this insight allows us to state the following lemma:
    
    \begin{lemma} \label{lemma:mss}
		Let $M_\vy$ be an MSS of $S_\vy$ and $\hy$ be an assignment that satisfies $M_\vy$. Then: (1) for every assignment $\hx$ such that $MustSat(\hx) \subseteq M_\vy$,  the assignment $(\hx, \hy)$ satisfies $F(\vx,\vy)$; and (2) for every assignment $\hx$ such that $MustSat(\hx) \supset M_\vy$, there is no $\hy'$ such that the assignment $(\hx, \hy')$ satisfies $F(\vx, \vy)$.
	\end{lemma}
	
	\begin{proof}
		(1) Since $\hy$ satisfies every clause $C_i|_\vy$ in $M_\vy$, it must be that $\hy$ also satisfies every clause in $MustSat(\hx)$. Therefore, for every clause $C_i$ in $F$, either $C_i|_\vx$ is satisfied by $\hx$ (and therefore $C_i|_\vy \not\in MustSat(\hx)$) or $C_i|_\vy$ is satisfied by $\hy$. Therefore $(\hx, \hy)$ satisfies $F(\vx, \vy)$.
		(2) Since $M_\vy$ is maximal, then in this case $MustSat(\hx)$ must be unsatisfiable. Therefore there is no $\hy'$ that can satisfy all clauses that $\hx$ does not already satisfy.
	\end{proof}
    
    Therefore, similarly to the analysis of the input component, we have:
    
    \begin{corollary}
    $F$ can be implemented by a decision list of length equal to the number of MSS of $S_\vy$, where each $f_i$ in the decision list is of size linear in the size of the specification.
    \end{corollary}
    
    \begin{proof}
    Construct $f_i(\vx)$ by taking the conjunction of all input clauses $C|_\vx$ such that $C|_\vy$ is not contained in the $i$-th MSS $M_i$. Then, $f_i(\vx)$ is satisfied exactly by those assignments $\hx$ such that $MustSat(\hx)$ is a subset of $M_i$.
Then, set the corresponding output assignment $\hy_i$ to an arbitrary satisfying assignment of $M_i$.
    \end{proof}
    
\begin{example} \label{ex:mss}
Let $F$, $S_\vx$ and $S_\vy$ be the same as in Example~\ref{ex:mfs}. $S_\vy$ has three MSS: $\{(y_1), (y_1 \lor \neg y_2), (y_2)\}$, $\{(\neg y_1), (y_1 \lor \neg y_2)\}$ and $\{(\neg y_1), (y_2)\}$. From these MSS we can construct a decision list implementing $F$ in the way described above. Note that some decisions in the list might be redundant:
\begin{align*}
&\texttt{if } (x_1 \lor x_2) \texttt{ then } (y_1 := 1; y_2 := 1) \\
&\texttt{else if } (x_1 \lor \neg x_2) \land (\neg x_1 \lor x_2) \texttt{ then }(y_1 := 0; y_2 := 0) \\
&\texttt{else if } (x_1 \lor \neg x_2) \land (x_2) \texttt{ then }(y_1 := 0; y_2 := 1)
\end{align*}
\end{example}
    
   Unlike the input component,  the output analysis does not require the specification to be realizable to produce the correct answer: for every input $\hx$ for which an output $\hy$ exists, $MustSat(\hx)$ will be contained in some MSS, and therefore will be covered by the decision list. On the other hand, we do not care about the case where an input $\hx$ has no corresponding output $\hy$. Note, however, that unlike the input component,  we do not have here a  simple graph structure that can be exploited to obtain the list of MSSs, and finding an MSS is clearly NP-hard. Therefore, it is unlikely for us to be able to efficiently identify instances where the number of MSS is polynomial.
    
    More importantly, however, is that taking into account only the  output component and ignoring the input component may also lead to a large decision list that includes many MSSs that would never be activated by an input. This fact emphasizes the drawbacks of independent synthesis of the components,
and  motivates the development of an algorithm that combines the input and output analyses to produce a decision list that is smaller than either of the ones produced by each analysis individually.

\subsection{Alternating between Input and Output Components} \label{sec:back_and_forth}

Our next goal is to combine the input and output analyses obtained so far into a synthesis procedure that constructs a decision list of length upper-bounded by the minimum among the number of MFS of the input clauses and the number of MSS of the output clauses. 
Due to the restrictions of the input analysis, if the specification is unrealizable the procedure terminates without producing a decision list. Extending the synthesis to unrealizable specifications is left for future work.
We first state the following lemma:

\begin{lemma} \label{lemma:mfs_mss}
If $F(\vx, \vy)$ is realizable, then for every MFS $M_\vx$ of $S_\vx$, $\Co(M_\vx) \subseteq M_\vy$ for some MSS $M_\vy$ of $S_\vy$.
\end{lemma}

\begin{proof}
For every MFS $M_\vx$, since $M_\vx$ is all-falsifiable, there exists an input assignment $\hx$ such that $Fals(\hx) = M_\vx$. Then, since $F$ is realizable, $MustSat(\hx) = \Co(M_\vx)$ is satisfiable, and therefore is contained in some MSS.
\end{proof}

Given an MFS $M_\vx$ for the input clauses, we say that an MSS $M_\vy$ for the output clauses \emph{covers} $M_\vx$ if $\Co(M_\vx) \subseteq M_\vy$. Lemma~\ref{lemma:mfs_mss} says that for every MFS $M_\vx$, there exists at least one MSS $M_\vy$ that covers $M_\vx$. Therefore, instead of producing a satisfying assignment for $\Co(M_\vx)$, we can produce a satisfying assignment for $M_\vy$. In fact, such satisfying assignment also takes care of every other MFS covered by $M_\vy$, making it unnecessary to generate them.

The above insight gives rise to Algorithm~\ref{alg:back_and_forth}, which we call the "Back-and-Forth" algorithm. In this algorithm, we maintain a list $L$ of MSSs that is initially empty. At every iteration of the algorithm, we produce a new MFS that is not covered by the MSSs already in $L$. Then, we find an MSS that covers this new MFS. If no such MSS exists, it means the specification is unrealizable, and so the algorithm emits an error message and terminates. Otherwise, we add this MSS to $L$. After all the  MFS have been covered, we construct a decision list from
the obtained list $L$ of MSS in the same way as described in Section~\ref{sec:output_analysis}: $f_i(\vx)$ is a formula that is satisfied exactly when $MustSat(\vx)$ is a subset of the $i$-th MSS, and the corresponding output assignment $\hy_i$ is a satisfying assignment for that MSS.

\begin{algorithm}[t]
\begin{algorithmic}[1]
\State initialize a list of MSSs $L$ to the empty list
\While{there are still MFS left to generate}
	\State $M_\vx \gets \text{MFS of $S_\vx$ not covered by any MSS in $L$}$ \label{line:mfs}
    \If{MSS $M_\vy \subseteq S_\vy$ covering $M_\vx$ exists} \label{line:mss}
    	\State add $M_\vy$ to $L$
    \Else
    	\State FAIL: specification is unrealizable
    \EndIf
\EndWhile
\State construct decision list from $L$
\end{algorithmic}
\caption{Back-and-Forth synthesis algorithm combining MFS and MSS analysis.}
\label{alg:back_and_forth}
\end{algorithm}

\begin{example}
Let $F$, $S_\vx$ and $S_\vy$ be the same as in Examples~\ref{ex:mfs} and~\ref{ex:mss}. In the first iteration, we generate the MFS $M^1_\vx = \{(x_1 \lor \neg x_2)\}$. Then, we expand $Co(M^1_\vx) = \{(y_1)\}$ into the MSS $M^1_\vy = \{(y_1), (y_1 \lor \neg y_2), (y_2)\}$ and add $M^1_\vy$ to $L$. Note that $M^1_\vy$ also covers, besides $M^1_\vx$, the MFS $\{(x_2), (\neg x_1 \lor x_2)\}$, and therefore this MFS will not need to be  generated. The only remaining MFS is $M^2_\vx = \{(x_1 \lor x_2), (x_2)\}$. $M^2_\vy = Co(M^2_\vx) = \{(\neg y_1), (y_1 \lor \neg y_2)\}$ is already an MSS, so we add it to $L$. Since all MFS have been covered, the procedure terminates. Note that we did not need to add the MSS $\{(\neg y_1), (y_2)\}$ to $L$, since no MFS is covered by this MSS. From $L$, we can now construct a decision list as described earlier:
\begin{align*}
&\texttt{if } (x_1 \lor x_2) \texttt{ then } (y_1 := 1; y_2 := 1) \\
&\texttt{else if } (x_1 \lor \neg x_2) \land (\neg x_1 \lor x_2) \texttt{ then }(y_1 := 0; y_2 := 0)
\end{align*}
\end{example}

\paragraph{Implementation details} The key steps of Algorithm~\ref{alg:back_and_forth} are the generation of the MFS $M_\vx$ in line~\ref{line:mfs} and the MSS $M_\vy$ in line~\ref{line:mss}. These steps are similar to the input and output analyses in Sections~\ref{sec:input_analysis} and~\ref{sec:output_analysis}. Since, however, we use communication between the input and output components, we have additional constraints on the MFS and MSS being generated. At each step the generated MFS must not be covered by the previously-generated MSSs, and the generated MSS must cover the most recently generated MFS.

While generating an arbitrary MFS can be done in polynomial time, we prove that adding the restriction that the MFS must not be covered by a previous MSS makes the MFS generation an NP-complete problem (see 
appendix
for proper theorem and proof). Therefore, we implement the MFS generation in the following way. First, we use a SAT solver as an NP oracle to find an (not-necessarily maximal) all-falsifiable subset of $S_\vx$ not covered by the previous MSSs. Then, we  extend this subset to an MFS by iterating over the remaining input clauses and at each step adding to the growing set a clause that does not conflict with the clauses already present in that set. This process of obtaining an MFS from $S_\vx$ is easier to implement when we use the conflict graph representation of $S_\vx$. Given $k$ previous MSSs $M_1, \ldots, M_k$ and the conflict graph $G = (V, E)$, we use the following SAT query to generate an all-falsifiable subset:

\begin{equation*}
\varphi\equiv\bigwedge^k_{i=1} \left(\bigvee_{C_j|_\vy~\in~S_\vy \setminus M_i} z_j\right) \land \bigwedge_{(C_i|_\vx, C_j|_\vx)~\in~E} (\neg z_i \lor \neg z_j)
\end{equation*}

We use variable $z_i$ to indicate whether clause $C_i|_\vx$ is present in the all-falsifiable subset. The first conjunction encodes that for every previous MSS, the subset must include a clause $C_j|_\vx$ not covered by that MSS. The second conjunction expresses that if two clauses conflict with each other, they cannot both be added to the subset. Note that whenever we generate a new MFS, we only need to add extra clauses of the first form
to this query, allowing us to employ incremental capabilities of SAT solvers.

After extending the subset produced by the SAT solver to an MFS $M_\vx$, we have to generate a new MSS $M_\vy$ that covers $M_\vx$. For that we use a partial MaxSAT solver as an oracle. In a partial MaxSAT problem, some clauses are set as hard clauses and others are set as soft clauses~\cite{ABL09}. The solver then returns an assignment that satisfies all hard clauses and the maximum possible number of soft clauses. We call the MaxSAT solver on the set of output clauses $S_\vy$, where the clauses in $\Co(M_\vx)$ are set as hard clauses, and all other clauses are set as soft clauses. This way, the MaxSAT solver is guaranteed to return a satisfiable set of clauses containing $\Co(M_\vx)$ and of maximum size. Since a satisfiable subset of maximum size is necessarily maximal, the satisfied clauses returned by the MaxSAT solver is an MSS, as desired.

\paragraph{Analysis and Correctness} Since exactly one new MFS and one new MSS are generated at every iteration, the number of iterations in Algorithm~\ref{alg:back_and_forth} is upper bounded by $\min(\#MFS, \#MSS)$. Yet, since Algorithm~\ref{alg:back_and_forth} does not generate redundant MFS and MSS, the number of iterations, and thus the size of the decision list, can be much smaller.

We now formalize and prove the correctness of Algorithm~\ref{alg:back_and_forth}.

\begin{lemma} \label{lemma:back_and_forth}
For a realizable specification $F(\vec{x}, \vec{y})$, let $\langle (f_1, \hy_1), \ldots, (f_k, \hy_k) \rangle$ be the decision list produced by Algorithm~\ref{alg:back_and_forth}. Then (1) For every $\hx$ such that $f_i(\hx) = 1$, $F(\hx, \hy_i) = 1$;
(2) For every $\hx$ there is at least one $i$ such that $f_i(\hx) = 1$.
\end{lemma}

\begin{proof}
(1) Let $M_\vy$ be the $i$-th MSS generated by the algorithm. Then, by construction, $f_i(\hx) = 1$ iff $MustSat(\hx) \subseteq M_\vy$, and $\hy_i$ is a satisfying assignment to $M_\vy$. Therefore, if $f_i(\hx) = 1$ then $\hy_i$ satisfies $MustSat(\hx)$, and so $(\hx, \hy_i)$ satisfies $F$.

(2) For every $\hx$, there exists an MFS $M_\vx$ such that $Fals(\hx) \subseteq M_\vx$. If $M_\vx$ was generated by the algorithm, then an MSS $M_\vy$ that covers $M_\vx$ was added to the MSS list. If $M_\vx$ was not generated by the algorithm, it must be because there was already a previously generated MSS $M_\vy$ that covers $M_\vx$. Either way, since $M_\vy$ covers $M_\vx$ and $Fals(\hx) \subseteq M_\vx$, $M_\vy$ covers $Fals(\hx)$. Therefore, the corresponding $f_i$ in the decision list is such that $f_i(\hx) = 1$.
\end{proof}

From Lemma \ref{lemma:back_and_forth}  we obtain the following corollary.

\begin{corollary}
Given a realizable specification $F(\vec{x}, \vec{y})$, the decision list produced by Algorithm~\ref{alg:back_and_forth} implements $F$.
\end{corollary}

It is worth noting that if the number of MFS is small as discussed in  Section~\ref{sec:input_analysis}, then purely enumerating MFS, as in Section~\ref{sec:input_analysis} can be theoretically faster than using Algorithm ~\ref{alg:back_and_forth}. That is because finding an MFS can be done in polynomial time, while Algorithm \ref{alg:back_and_forth} requires calls to a SAT and MaxSAT solvers. In practice, however, we observed that  the Back-and-Forth algorithm often avoids a large number of redundant MFS, which makes up for the extra complexity in generating each MFS. Still, for specifications that are known to have a small number of MFS,  restriction to the analysis of the input component as in Section~\ref{sec:input_analysis}  can be sufficient.

\subsection{Partitioning the Specification into Distinct Output Variables} \label{sec:cc}

Some of the cases in the back-and-forth analysis which cause the number of MFS or MSS to be exponential can be simplified by partitioning the specification into sets of clauses that do not share output variables. As an example, consider the specification for the identity function:
$$F(\vx, \vy) = (x_1 \leftrightarrow y_1) \land \ldots \land (x_k \leftrightarrow y_k)$$
or in a CNF form:
$$F(\vx, \vy) = (\neg x_1 \lor y_1) \land (x_1 \lor \neg y_1) \land \ldots \land (\neg x_k \lor y_k) \land (x_k \lor \neg y_k)$$

It is easy to see that both the number of MFS and MSS for this formula are $2^k$. Each output variable, however, does not appear in the same clause with other output variables. Therefore, we can consider each pair $(\neg x_i \lor y_i) \land (x_i \lor \neg y_i)$ of clauses as a separate specification and synthesize it independently as a decision list of size 2. As such, the total number of MFS and MSS grow linearly with $k$.  

Therefore we propose the following preprocessing step.

\begin{enumerate}
\item Given the specification $F$, construct a graph with a vertex for each clause and an edge between two vertices iff the corresponding clauses share an output variable.
\item Separate the graph into connected components $\mathbb{C}_1,\ldots,\mathbb{C}_k$. Note that the $\mathbb{C}_i$ are completely disjoint in terms of output variables.
\item For every $\mathbb{C}_i$, define a sub-specification $F_i$ by taking only the clauses in $F$ whose corresponding vertex is in $\mathbb{C}_i$.
\item  Call Algorithm~\ref{alg:back_and_forth} for each specification $F_i$. This gives us a decision list $D_i$ for $F_i$ that decides on an assignment for only the output variables in $F_i$.
\end{enumerate}

Since the $F_i$ have disjoint sets of output variables, every $D_i$ decides on an assignment for a different partition of output variables. Therefore, given an input $\hx$ we can produce a corresponding output $\hy$ by simply evaluating each $D_i$ independently on $\hx$ and combining the results.

\section{Experimental Evaluation}\label{sec:exp}

In order to evaluate the performance of the Back-and-Forth synthesis algorithm, we ran the algorithm on benchmarks from the 2QBF track of the QBFEVAL'16 QBF-solving competition~\cite{QBFEval}. This track is composed of QBF benchmarks of the form $\forall \vec{x} . \exists \vec{y} . F(\vec{x}, \vec{y})$, where $F$ is a CNF formula. We can see these benchmarks as synthesis problems asking if we can synthesize a Skolem function for the existential variables in terms of the universal variables such that the formula $F$ is satisfied. For this experimental evaluation we  used only those benchmarks that are realizable, since adjusting the Back-and-Forth algorithm to handle unrealizable benchmarks is future work. The benchmarks can be classified into seven families: \textsc{MutexP} (7 instances), \textsc{Qshifter} (6 instances), \textsc{RankingFunctions} (49 instances), \textsc{ReductionFinding} (34 instances), \textsc{SortingNetworks} (22 instances), \textsc{Tree} (5 instances) and \textsc{FixpointDetection} (93 instances). Because benchmarks in the same family tend to have similar properties, it makes sense to evaluate performance over each family, rather than over specific instances.

We compared the running time of the Back-and-Forth algorithm on these benchmarks with three state-of-the-art tools that employ different synthesis approaches: the CDCL-based CADET~\cite{RS2016}, the ROBDD-based RSynth~\cite{TV17}, and the CEGAR-based BFSS~\cite{ACGKS2018}. Since the Back-and-Forth algorithm, CADET and RSynth are all sequential algorithms, to ensure fair comparison of computational effort, the version of BFSS used was compiled with the MiniSAT SAT solver~\cite{ES03} instead of the parallelized UniGen sampler used in~\cite{ACGKS2018}. We leave for future work the exploration of performance of the different tools in a parallel scenario.

Our implementation of the Back-and-Forth algorithm used the Glucose SAT solver~\cite{AS09}, based on MiniSAT, and the Open-WBO MaxSAT solver~\cite{MML14}. The implementation also used the partitioning described in Section~\ref{sec:cc}.
All experiments were executed in the DAVinCI cluster at Rice University, consisting of 192 Westmere nodes of 12 processor cores each, running at 2.83 GHz with 4 GB of RAM per core, and 6 Sandy Bridge nodes of 16 processor cores each, running at 2.2 GHz with 8 GB of RAM per core.  Our algorithm has not been parallelized, so the cluster was solely used to run multiple experiments simultaneously. Each instance had a timeout of 8 hours.

Figure~\ref{fig:families} shows for each family the percentage of instances each tool was able to solve in the time limit. We can divide the results into three parts:

\begin{figure}[!t]
\centering
\includegraphics[width=\columnwidth]{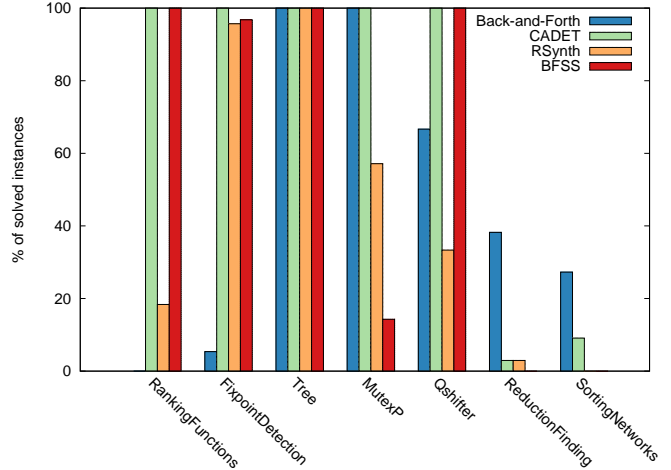}
\caption{Percentage of instances solved by each synthesis algorithm for each of the benchmark families.}
\label{fig:families}
\end{figure}

In the \textsc{RankingFunctions} and \textsc{FixpointDetection} families the Back-and-Forth algorithm timed out on almost all instances, only being able to solve the easiest instances of \textsc{FixpointDetection}. CADET, on the other hand, performed very well, being able to solve all instances. RSynth and BFSS also outperformed the Back-and-Forth algorithm, although they did not perform as well as CADET.

The \textsc{Tree}, \textsc{MutexP}, and \textsc{Qshifter}  families
had almost all instances solved by the Back-and-Forth algorithm in under 45 seconds (except for the two hardest instances of \textsc{Qshifter}, which timed out), in many cases outperforming RSynth or BFSS. Even so, CADET still performed the best in these classes, solving all instances faster than our algorithm.

Lastly, \textsc{ReductionFinding} and \textsc{SortingNetworks} seem to be the most challenging families for existing tools, with CADET only being able to solve two instances in total, RSynth one, and BFSS none. In contrast, our Back-and-Forth algorithm solved 13 cases in \textsc{ReductionFinding} and 6 in \textsc{SortingNetworks}. Furthermore, as can be seen in Figure~\ref{fig:instances}, every instance that was solved by other tools  was also solved by the Back-and-Forth algorithm, which was faster by over an order of magnitude.


\begin{figure}[!t]
\centering
\includegraphics[width=\columnwidth]{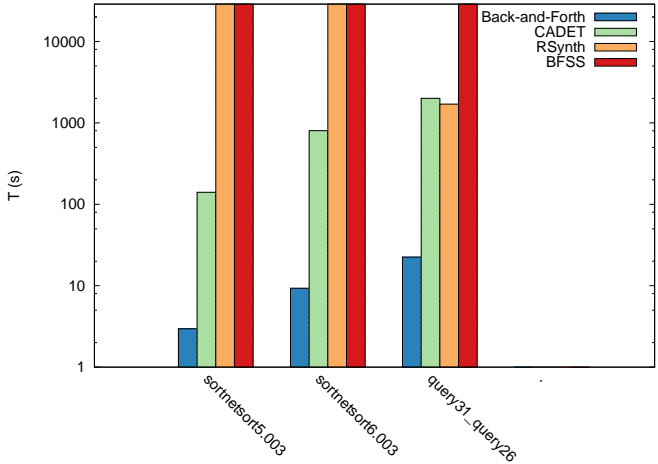}
\caption{Running time, in seconds, of each synthesis algorithm on instances of the \textsc{ReductionFinding} and \textsc{SortingNetworks} families that were solved by at least one algorithm besides Back-and-Forth. Bars of maximum height indicate the algorithm timed out on the benchmark.}
\label{fig:instances}
\end{figure}

In summary, the Back-and-Forth algorithm performed competitively in 5 out of 7 families, and was strictly superior in 2 out of 7 families.
Due to the difficulty of analyzing CNF formulas, the exact reason why the algorithm performs well in these particular families and not in others remains an open question, to be explored in future work. Still, the results suggest that the Back-and-Forth algorithm can serve as a good complement to modern synthesis tools, performing well exactly in the cases in which these tools struggle the most, and therefore it would be a good candidate for membership in a portfolio of synthesis algorithms. 

\section{Discussion}\label{sec:discussion}

%

A recurrent observation in recent
evaluations~\cite{JSCTA15,ACJS17,ACGKS2018,TV17} of Boolean functional
synthesis tools has been that no single tool or algorithm dominates
the others in all classes of benchmarks.  To build industry-strength
Boolean functional solvers, it is therefore inevitable that a
portfolio approach be adopted.  Since decomposition-based techniques
(beyond factored specifications) have not been used in
existing tools so far, our original motivation was to develop a
decomposition-centric framework for Boolean functional synthesis that
complements (rather than dominates) the strengths of existing tools.
As our experiments with the Back-and-Forth algorithm show, we have
been able to take the first few steps in this direction by
successfully solving some classes of benchmarks that state-of-the-art
tools choke on.  While we have tried to understand features of these
benchmarks that make them particularly amenable to our technique,
a lot more work remains to be done to elucidate this relation clearly.

Yet another motivation for exploring a decomposition-centric synthesis approach was to be able to generate Skolem functions in a format that lends itself to easy independent validation by domain experts.
Interestingly, despite the singular importance of this aspect, it has been largely ignored by existing Boolean functional synthesis tools, most of which construct a circuit representation of the function using an acyclic-graph data structure such as an ROBDD or an And-Inverter Graph. While these are known to be efficient representations of Boolean functions, they are not amenable to easy validation by a domain expert, especially when their sizes are large, often requiring a satisfiability solver to check that the generated Skolem functions indeed satisfy the specifications.
Synthesizing functions as decision lists is a natural and well-studied choice for meeting this objective. Along with each decision in the decision list, we can also identify the clauses that contribute to the generation of the outputs (these are clauses whose input components are falsified by the decision), thereby providing clues about which part of the specification is responsible for the outputs generated in a particular branch of the decision-list representation. Our work shows that decomposition-based techniques lend themselves easily to such representations.

In order to be consistent with performance comparison experiments
reported in the literature, all specifications used in our evaluation
were prenex CNF (PCNF) formulas taken from the QBFEVAL'16 benchmark suite.  While
this certainly presents challenging instances of Boolean functional
synthesis, PCNF is not a natural choice of representing specifications
in several important application areas.  For example, the industry
standard (IEC 1131-3) for reactive programs for programmable logic
controllers (PLC) includes a set of languages that allow the user to
specify combinations of outputs based on different combinations of
input conditions. The same is also true in the specification of
several bus protocols like the VME Bus or AMBA Bus.
Scenario-based
specifications such as these are much more amenable to our decomposition-based
approach, since there is a natural separation of input and output
components of the specification.  In addition, with such
specifications, it is meaningful to analyze the structure of
dependence between the input and output components, and exploit
structural properties (viz. the size of the MIS in the conflict graph
as explained in Section~\ref{sec:amicable}) in synthesis.  We believe
that as we look beyond PCNF representations of specifications,
techniques like those presented in this paper will be even more
useful in a portfolio approach to synthesis.

In our experimental evaluation, we chose CADET as a representative of the state-of-the-art on Boolean synthesis stemming from the QBF community. This is due to its focus on 2QBF (which suffices for Boolean synthesis of realizable specifications) and its performance on recent QBFEVAL competitions. Another certifying QBF solver, CAQE~\cite{RT2015}, uses techniques that are similar to the clause splitting used in our algorithm. But CAQE targets QBF instances with arbitrary quantifier alternation, requiring additional mechanisms for handling these cases, and furthermore does not perform the same analysis as here, based on MFS and MSS. Due to their similarities, it would be interesting to perform a comparison between the two algorithms in the future.

Finally, the techniques presented in this work are clearly not the only ways to achieve synthesis via decomposition, and there exists scope for significant innovation and creativity, both in the manner in which a specification is decomposed, and in the way the decomposition is exploited to arrive at an efficient synthesis algorithm. One example lies in identifying algorithms for sequential decomposition, as presented in~\cite{FLOV18}, which are applicable to a synthesis context. In summary, synthesis based on input-output decomposition presents uncharted territory that deserves systematic exploration in order to complement the strengths of existing synthesis tools.

\section*{Acknowledgment}

We thank Assaf Marron for useful discussions, and the anonymous reviewers for their suggestions.



%
\bibliographystyle{abbrv}
\bibliography{BDDSeqDec}

\clearpage

\section*{Appendix}

\subsection{On Synthesis via Sequential Decomposition}

As a first attempt for synthesis via decomposition, we explored the use of a decomposition method called \emph{ sequential decomposition}, described in \cite{FLOV18},  for the purpose of synthesis. In sequential decomposition the specification $F$ is split into 
input and output parts by adding intermediate fresh variables $\vz = (z_1, \ldots, z_k)$ 
defining a domain $Z$
that serves to communicate between the input domain $X$ and the output domain $Y$. This intermediate domain $Z$ should be introduced in such a way as to preserve exactly every input/output pair in $F$.  In addition, to preserve the independence of the two parts, as described in  \cite{FLOV18},  we would like each part to  be synthesized independently and then recomposed into an implementation for the entire specification.  Therefore we define the following.
    
	\begin{definition} \label{def:decomposition}
		Let $F(\vx,\vy), F_1(\vx,\vz),F_2(\vz,\vy)$ be Boolean formulas. Then $(F_1,F_2)$ is called a \emph{good decomposition} of $F$ if
        \begin{inparaenum}
        \item $F(\vx,\vy)\equiv \exists \vz . (F_1(\vx,\vz)\land F_2(\vz,\vy))$; and \label{item:decomp}
        \item for every input $\hx \in Dom(F)$, $Img_{\hx}(F_1) \subseteq Dom(F_2)$. \label{item:recomp}
        \end{inparaenum}
	\end{definition}
    
    Property~(\ref{item:decomp}) guarantees that for every input assignment $\hx$ and output assignment $\hy$, $(\hx, \hy)$ satisfies $F$ if and only if there exists an intermediate assignment $\hz$ such that $(\hx, \hz)$ satisfies $F_1$ and $(\hz, \hy)$ satisfies $F_2$. Property~(\ref{item:recomp}) guarantees that for all implementations $g_1$ of $F_1$ and $g_2$ of $F_2$, their composition $g_2 \circ g_1$ is well-defined and is an implementation of $F$. Such a decomposition attains a complete separation of the inputs and outputs of $F$, in the sense that no direct knowledge of the output variables is necessary to synthesize $F_1$, nor of the input variables to synthesize $F_2$.

       We  now state the following theorem describing synthesis by sequential decomposition. 
        
		\begin{theorem}\label{thm:StrongDec}
			Let $F(\vec{x}, \vec{y})$ be a specification where $\vec{x}$ are the input variables and $\vec{y}$ are the output variables. If $F_1(\vec{x}, \vec{z})$ and $F_2(\vec{z}, \vec{y})$ form a good decomposition of $F$, then for every implementation $g_1$ of $F_1$ and $g_2$ of $F_2$, $g_2 \circ g_1$ implements $F$.
		\end{theorem}
        
       \begin{proof}
			Since $F(\vx,\vy)\equiv \exists \vz . (F_1(\vx,\vz)\land F_2(\vz,\vy))$, we have that if $(\hx,\hz)$ satisfies $F_1$ and $(\hz,\hy)$ satisfies $F_2$ then $(\hx,\hy)$ satisfies $F$.
			Let $g_1:X \to Z$ and $g_2:Z \to Y$ be implementations of $F_1$ and $F_2$, respectively.
			Let $\hx \in Dom(F)$. Since $(F_1, F_2)$ is a good decomposition of $F$, $\hx \in Dom(F_1)$.
            Since $g_1$ is an implementation of $F_1$, $(\hx,g_1(\hx))$ satisfies $F_1$. Furthermore, since $\hx \in Dom(F)$, $Img_{\hx}(F_1) \subseteq Dom(F_2)$. Then, since $g_1(\hx) \in Img_\hx(F_1)$, $g_1(\hx)\in Dom(F_2)$. Therefore, $(g_1(\hx),g_2(g_1(\hx))$ satisfies $F_2$. Since $(\hx,g_1(\hx))$ satisfies $F_1$ and $(g_1(\hx),g_2(g_1(\hx))$ satisfies $F_2$, then $(\hx, g_2(g_1(\hx)))$ satisfies $F$. Therefore, $g_2 \circ g_1$ is an implementation of $F$.
		\end{proof}
			
    Theorem~\ref{thm:StrongDec} describes a clean condition for a decomposition that allows synthesis of each component independently. 
We 
now provide a concrete example that follows 
this framework.
Specifically we introduce a type of decomposition that satisfies the requirements of a good decomposition according to Definition~\ref{def:decomposition}. This decomposition can be applied to all specifications $F(\vx, \vy)$ in CNF, and is based on the same concepts used in the input and output analyses in Section~\ref{sec:amicable}.

The \emph{CNF decomposition} of $F = C_1 \land \ldots \land C_k$ is  a pair $(F_1,F_2)$ where
		\begin{align*}
		&F_1(\vec{x}, \vec{z}) = \bigwedge^k_{i=1} (\neg C_i|_{\vec{x}} \leftrightarrow z_i) \\
		&F_2(\vec{z}, \vec{y}) = \bigwedge^k_{i=1} (\neg z_i \lor C_i|_{\vec{y}}) \equiv \bigwedge^k_{i=1} (z_i \rightarrow C_i|_{\vec{y}})
		\end{align*}
        The idea behind the CNF decomposition is to focus on the clauses that can be true/false according to the assignment for the input component. This leads to a natural and very simple decomposition: in fact,  $F_1$ is already a function from $\vx$ to $\vz$ on its own, and hence the synthesis of $g_1$ is trivial - just assign every $z_i$ to $\neg C_i|_\vx$. Intuitively, this decomposition works by grouping assignments of the input variables into individual $z$ variables. Specifically, note that $z_i$ is only assigned to true if that is absolutely necessary, that is, when  $C_i|_{\vec{x}}$ is not satisfied and we must satisfy $C_i|_{\vec{y}}$ instead. As such, we abstract away  all the assignments that make the same $z_i$ variables true. Therefore, we only need to concern ourselves with synthesizing $g_2$ from $F_2$. 
  
  We now prove that this decomposition meets the criteria for a good decomposition according to Theorem~\ref{thm:StrongDec}.

	\begin{theorem}
		If $F_1$ and $F_2$ are given by the CNF decomposition of a CNF formula $F$, then $(F_1,F_2)$ is a good decomposition of $F$.
	\end{theorem}
    
	


\begin{proof}
We first prove that $F(\vx, \vy) \equiv \exists \vz . F_1(\vx, \vz) \land F_2(\vz, \vy)$.

\begin{align*}
\exists \vz . F_1(\vx, \vz) \land F_2(\vz, \vy)
	&\equiv \exists \vz . \bigwedge^k_{i=1} (\neg C_i|_\vx \leftrightarrow z_i) \land \bigwedge^k_{i=1} (z_i \rightarrow C_i|_\vy) \\
    &\equiv \bigwedge^k_{i=1} (\exists z_i . (\neg C_i|_\vx \leftrightarrow z_i) \land (z_i \rightarrow C_i|_\vy)) \\
    &\equiv \bigwedge^k_{i=1} (\neg C_i|_\vx \rightarrow C_i|_\vy) \\
    &\equiv \bigwedge^k_{i=1} (C_i|_\vx \lor C_i|_\vy) \\
    &\equiv \bigwedge^k_{i=1} C_i \\
    &\equiv F(\vx, \vy)
\end{align*}

Next, we prove that $Img_\hx(F_1) \subseteq Dom(F_2)$ for every input $\hx \in Dom(F)$. Assume $\hx \in Dom(F)$, that is, there exists $\hy$ such that $F(\hx, \hy) = 1$. Since $F_1(\vx, \vz) = (\neg C_1|_\vx \leftrightarrow z_1) \land \ldots \land (\neg C_k|_\vx \leftrightarrow z_k)$, there is a unique $\hz$ that satisfies $F_1$ for $\hx$. This $\hz$ is the only element of $Img_\hx(F_1)$, therefore we only need to prove that $\hz \in Dom(F_2)$. But by construction of $F_2$, the same $\hy$ that satisfies $F$ for $\hx$ also satisfies $F_2$ for $\hz$. Therefore, $\hz \in Dom(F_2)$.
\end{proof}


In order to evaluate the possibility of using CNF decomposition for synthesis, we compared synthesis performance of the decomposed formula with the performance for the original formula. Since the decomposition itself and synthesis of $F_1$ are straightforward and can be performed in linear time, we focused on comparing the synthesis of $F_2$ with the direct synthesis of $F$.

One important detail is that, even if $F$ is realizable, the $F_2$ produced by CNF decomposition might not be. As a consequence, we have to use a synthesis tool that is able to handle unrealizable specifications, such as RSynth. We ran the experiments using the same experimental setting as Section~\ref{sec:exp}.

\begin{figure}[!t]
\centering
\includegraphics[width=\columnwidth]{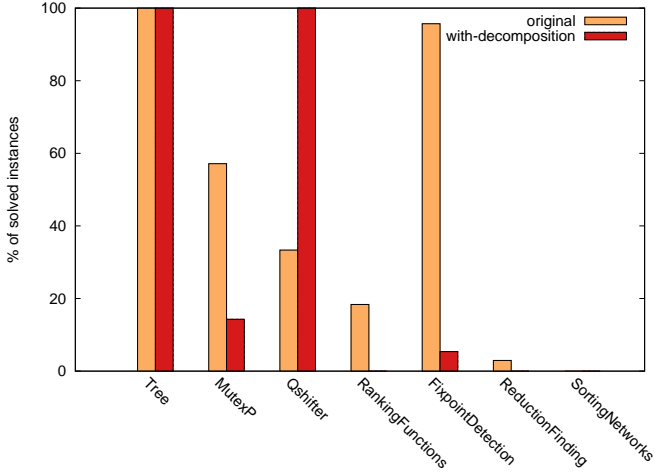}
\caption{Percentage of instances solved by the RSynth synthesis tool in the time limit when synthesizing the original specification vs. the decomposed version.}
\label{fig:decomp}
\end{figure}

As can be seen from Figure~\ref{fig:decomp}, in most cases the performance of RSynth worsened when using CNF decomposition. For most families, when applied to the decomposed version the tool was able to solve only a much smaller subset of the benchmarks. Furthermore, instances that could be solved took significantly longer in general, sometimes up to 10,000 times more. The only exception occurred in the \textsc{Qshifter} family, in which decomposition scaled significantly better and allowed solving all the instances.

The conclusion that we draw from these experiments is that CNF decomposition is not amenable to synthesis, and in many cases can worsen synthesis performance considerably. Part of the reason is likely 
because the synthesis of $F_2$ has no access to information about $F_1$, possibly performing superfluous work to produce a correct function for assignments of $\vz$ that are never produced by the input component. This insight lead us to pursue the approach presented in Section~\ref{sec:amicable}, that relies on sharing information between the two separated components.
It might still be possible to benefit from sequential decomposition for synthesis, but better decomposition strategies are clearly necessary.



\subsection{Construction of the Decision List from List of MFS/MSS}

In Section~\ref{sec:amicable} we implement a specification $F(\vx, \vy)$ in CNF by a decision list constructed from either a list of MFS of the input clauses or a list of MSS of the output clauses. Now we show how to perform this construction in detail.

Let $F(\vx, \vy)$ be a CNF formula, and $S_\vx$ and $S_\vy$ be, respectively, the set of input and output clauses of $F$.
Let $L_\vx$ be a list of MFS of $S_\vx$ and $L_\vy$ be the list of MSS of $S_\vy$. 

We first show how we can use $L_\vx$ to construct a decision list of the form \texttt{if $f_1(\vx)$ then $\hy_1$ else if $f_2(\vx)$ then $\hy_2$ else \ldots else $\hy_{k_\vx}$}, where $k_\vx$ is the length of $L_\vx$. Moreover, the size of each $f_i(\vx)$ is linear to the specification $F$.

For every $1 \leq i \leq k_\vx$ let $M^i_\vx$ be the $i$-th MFS. We define
$$f_i(\vx) = \bigwedge_{C|_\vx \in S_\vx \setminus M^i_\vx} C|_\vx$$
and $\hy_i$ to be a satisfying assignment to $\Co(M^i_\vx)$. Note that $f_i$ is satisfied exactly by those input assignments $\hx$ that satisfy every clause $C$ such that $C|_\vx \not\in M^i_\vx$, which means that $Fals(\hx) \subseteq M^i_\vx$. Meanwhile, $\hy_i$ satisfies every clause $C$ such that $C|_\vx \in M^i_\vx$. As a consequence, if $f_i(\hx) = 1$ then $(\hx, \hy_i)$ satisfies every clause in the CNF, and therefore satisfy $F$.

Next we show how we can use $L_\vy$ to construct a similar decision list of size $k_\vy$  --  the length of $L_\vy$.
For every $1 \leq i \leq k_\vy$ let $M^i_\vy$ be the $i$-th MSS. We define
$$f_i(\vx) = \bigwedge_{C|_\vy \in S_\vy \setminus M^i_\vy} C|_\vx$$
and $\hy_i$ to be a satisfying assignment to $M^i_\vx$. Note that $f_i$ is satisfied exactly by those input assignments $\hx$ that satisfy every clause $C$ such that $C|_\vy \not\in M^i_\vy$, which means $MustSat(\hx) \subseteq M^i_\vy$. Meanwhile, $\hy_i$ satisfies every clause $C$ such that $C|_\vy \in M^i_\vy$. As a consequence, if $f_i(\hx) = 1$ then $(\hx, \hy_i)$ satisfies every clause in the CNF, and therefore satisfy $F$.

Therefore, as long as $L_\vx$ or $L_\vy$ contain enough MFS or MSS respectively to cover every possible input assignment $\hx$, something that is guaranteed by the algorithms in Section~\ref{sec:amicable}, the generated decision list is a correct implementation of $F(\vx, \vy)$.

\subsection{NP-completeness of Constrained MFS Generation}

An essential component of the Back-and-Forth algorithm in Section~\ref{sec:back_and_forth} is the generation of an MFS that is not covered by the previous MSSs. We can formulate this problem in the following decision problem called \textsc{Hitting-MFS-Among-MSS}, due to its resemblance to the Hitting-set problem. We say that and MFS $M$ is \textit{covered by a list} $L$ of MSS if there exists an MSS $M'\in L$, such that $M$ is covered (in the sense of Section \ref{sec:amicable}) by $M'$.

\begin{problem} \label{prob:cmfsgen}
Let $F(\vx, \vy)$ be a CNF formula where $S_\vx$ is the set of input clauses and $S_\vy$ is the set of output clauses. Let $L=(M^1_\vy, \ldots, M^k_\vy)$, where each $M^i_\vy\subseteq S_\vy$, be a collection of MSSs of the output clauses $S_\vy$. Does there exist an MFS $M_\vx \subseteq S_\vx$ of the input clauses such that $M$ is not covered by $L$?
\end{problem}

We next show that \textsc{Hitting-MFS-Among-MSS} is NP-complete. Thus our choice of using a SAT solver in the Back-and-Forth algorithm to find the MFS is natural.

\begin{theorem}
\textsc{Hitting-MFS-Among-MSS} is NP-complete.
\end{theorem}

\begin{proof}
First note that \textsc{Hitting-MFS-Among-MSS} is clearly in \textsc{NP}, since given $M_\vx \subseteq S_\vx$, we can verify in polynomial time whether $M_\vx$ is an MFS and whether $M_\vx$ is covered by one of the $M^i_\vy$. We next prove that \textsc{Hitting-MFS-Among-MSS} is \textsc{NP}-hard by a reduction from 3SAT.

Let $\varphi(\vx) = D_1 \land \ldots \land D_k$ be a 3CNF formula where every $D_i$ is a clause. We  construct a CNF formula $F(\vx, \vy)$ and MSSs $L=(M^1_\vy, \ldots, M^k_\vy)$ such that there exists an MFS that is not covered by $L$  iff $\varphi$ is satisfiable.

We first construct the set of input clauses $S_\vx$ in the following way. For every clause $D_i$, let $SatAssign(D_i)$ be the set of all partial assignments to the variables of $D_i$ that satisfy $D_i$. Since $D_i$ has at most three literals, $|SatAssign(D_i)|$ is at most $7$. 
Next we construct a fresh input clause from every such assignment. Note that for every assignment $\tau$ we can construct a clause $C_\tau$ over the assigned variables that is falsified exactly by $\tau$.
For example if $\tau=(x_1=true,x_2=false,x_3=true)$ then the clause $C_\tau=(\neg x_1 \vee x_2 \vee \neg x_3)$ is falsified by $\tau$. Let $S^i_\vx = \{C_\tau \mid \tau \in SatAssign(D_i)\}$ be the set of input clauses falsified by the satisfying assignments of clause $D_i$. Let $S_\vx = \bigcup^k_{i=1} S^i_\vx$.

\begin{lemma} \label{lemma:mfsconstruct}
If $\hx$ is a satisfying assignment of $\varphi$, then the set $M_\vx \subseteq S_\vx$ of input clauses falsified by $\hx$ is an MFS.
\end{lemma}

\begin{proof}
Let $M_\vx$ be the set of input clauses falsified by $\hx$. By definition, $M_\vx$ is falsifiable, since it is falsifiable by $\hx$. We now prove that $M_\vx$ is maximal. For every clause $D_i$ there is a unique partial assignment $\tau_i$ over the variables in $D_i$ that is consistent with $\hx$. Since $C_{\tau_i} \in S_\vx$ and $\tau_i$ is consistent with $\hx$, $\hx$ falsifies $C_{\tau_i}$, and therefore $C_{\tau_i} \in M_\vx$. Then, as $M_\vx$ already has one such $\tau_i$ for every clause $D_i$, and all other satisfying assignments of $D_i$ conflict with $\tau_i$, $M_\vx$ is maximal.
\end{proof}

We next construct the set of output clauses $S_\vy$ in the following way. We first introduce one $y_i$ variable for every clause $D_i$. We then construct for every $D_i$ two clauses: $A_i = (y_i)$ and $B_i = (\bigvee_{j \neq i} \neg y_j)$. Let $M^i_\vy = \bigcup_{\ell \neq i} \{A_\ell, B_\ell\}$ for every $i\leq k$ and let $L=(M^1_\vy,\cdots,M^k_\vy)$.  Let $S_\vy = \bigcup^k_{i=1} \{A_i, B_i\}$.

\begin{lemma} \label{lemma:mssconstruct}
For $1 \leq i \leq k$, $M^i_\vy$ is an MSS of $S_\vy$.
\end{lemma}

\begin{proof}
First, we prove that $M^i_\vy$ is satisfiable. Note that $M^i_\vy$ has two types of clauses: $A_\ell = y_\ell$, for $\ell \neq i$, and $B_\ell = \bigvee_{j \neq \ell} \neg y_j$, for $\ell \neq i$. Therefore all $A_\ell$ can be satisfied by setting all $y_\ell$ for $\ell \neq i$ to true and all  $B_\ell$ can be satisfied by setting $y_i$ to false. Therefore, $M^i_\vy$ is satisfiable by setting $y_i$ to false and all other output variables to true.

To prove that $M^i_\vy$ is maximal, note that the only two output clauses that are missing from $M^i_\vy$ are $A_i = y_i$ and $B_i = \bigvee_{j \neq i} \neg y_j$. Note that for the $A_\ell$ clauses to be satisfied, all output variables other than $y_i$ must be set to true. Therefore, we cannot add $B_i$, since it would require one of these variables to be set to false. At the same time, for the $B_\ell$ clauses to be satisfied, since all other output variables are set to true, then $y_i$ must be set to false. Therefore, we cannot add $A_i$ without making the set unsatisfiable. As a consequence, $M^i_\vy$ is maximal.
\end{proof}

Lastly, we construct $F(\vx, \vy)$ by concatenating the input and output clauses in the following way. For every clause $D_i$ and assignment $\tau \in SatAssign(D_i)$, create two clauses $(C_\tau \lor A_i)$ and $(C_\tau \lor B_i)$.

We conclude with the following lemma.

\begin{lemma} \label{lemma:reduction}
There exists an MFS $M_\vx \subseteq S_\vx$ of the input clauses that is not covered by $L$ iff $\varphi$ is satisfiable.
\end{lemma}

\begin{proof}
($\leftarrow$) Assume $\varphi$ is satisfiable. Let $\hx$ be a satisfying assignment. Then, for every clause $D_i$ there must be a unique partial assignment $\tau_i$ consistent with $\hx$ that satisfies $D_i$. Then, the set $M_\vx = \bigcup^k_{i=1} \{C_{\tau_i}\}$ is the MFS defined in Lemma~\ref{lemma:mfsconstruct}. Now, recall that by the construction of $F$, that there must be a clause $(C_{\tau_i} \lor A_i)$ in $F$ for every 3CNF clause $D_i$. Therefore, $A_i \in \Co(M_\vx)$, for every $i$. Since $A_i \not\in M^i_\vy$, $M_\vx$ is not covered by $M^i_\vy$ for every $i$. Therefore, $M_\vx$ is an MFS that is not covered by $L$.

($\rightarrow$) Assume that there exists an MFS $M_\vx \subseteq S_\vx$ of the input clauses that is not covered by $L$. Since $M_\vx$ is an MFS, there must be an assignment $\hx$ that falsifies every clause $C_\tau \in M_\vx$. Every such $\tau$ must be consistent with $\hx$. For every $i$, since $M_\vx$ is not covered by $M^i_\vy$, $M_\vx$ must include an input clause $C_{\tau_i}$ originating from clause $D_i$. Since $\tau_i$ is consistent with $\hx$ and a satisfying assignment to $D_i$, $\hx$ must satisfy $D_i$. Therefore, $\hx$ satisfies every clause in $\varphi$, and therefore $\varphi$ is satisfiable.
\end{proof}

Therefore, our construction is a reduction from 3SAT to Problem~\ref{prob:cmfsgen}. Since the construction is polynomial, this proves that \textsc{Hitting-MFS-Among-MSS} is NP-complete as required.

\end{proof}

\subsection{Optimality of the Back-and-Forth Algorithm}

\begin{theorem}
Let $F(\vx, \vy)$ be a realizable specification. Let $k$ be the size of the smallest decision list that implements $F$. Then, there is an execution of the Back-and-Forth algorithm that produces a decision list of size $k$.
\end{theorem}

\begin{proof}
Let $L$ be a decision list of minimum size $k$. Assume without loss of generality that the $f_i(\vx)$ formulas partition the space of assignments of $\vx$. Note that every assignment $\hy$ output by the decision list satisfies a set of clauses, which must be contained in an MSS.

First, we prove that no two $\hy_i$ in $L$ share an MSS. Assume there is $i$ and $j$ such that $\hy_i$ and $\hy_j$ share an MSS $M_\vy$. Then, we can replace both $\hy_i$ and $\hy_j$ by a satisfying assignment $\hy$ of $M_\vy$, and this decision list remains an implementation of $F$. But now, since the $i$-th and $j$-th decisions produce the same output, we can merge them into a single decision with $f(\vx) = f_i(\vx) \lor f_j(\vx)$, and output $\hy$. The resulting decision list would have size $k - 1$ and would be a correct implementation of $F$, contradicting the minimality of $L$. Therefore, no $\hy_i$ and $\hy_j$ can share an MSS.

Now, choose for every $\hy_i$ an MSS $M^i_\vy$ containing the clauses satisfied by $\hy_i$. By the previous paragraph, these MSS must be all distinct. Let $L'$ be the decision list constructed from this list of MSS, as in Section~\ref{sec:back_and_forth}. Since there are $k$ MSS, $L'$ has length $k$. Note that $L'$ is also an implementation of $F$, since every input that satisfies the original $f_i(\vx)$ also satisfies the $f'_i(\vx)$ obtained from $M^i_\vy$, and the output $\hy'_i$ obtained from $M^i_\vy$ satisfies a superset of the clauses satisfied by the original $\hy_i$. Therefore, we only need to prove that this list of MSS can be constructed by the Back-and-Forth algorithm.

To prove that, we need to show that for every MSS $M^i_\vy$ there is an MFS $M^i_\vx$ such that:

\begin{enumerate}
\item $M^i_\vx$ is covered by $M^i_\vy$, and
\item $M^i_\vx$ is not covered by $M^j_\vy$, for all $j < i$.
\end{enumerate}

If these requirements are satisfied, then the list of MFS $M^1_\vx, \ldots, M^k_\vx$ can be produced by the Back-and-Forth algorithm. For the purpose of reaching a contradiction, assume that there is no such list of MFS. Then, there is some MSS $M^i_\vy$ such that, for every MFS covered by $M^i_\vy$, there is some MSS $M^j_\vy$ for $j < i$ that covers that MFS. But this means that every input covered by $M^i_\vy$ is already covered by a previous MSS. Therefore, $M^i_\vy$ can be removed from the list of MSS, which contradicts the minimality of $L'$.

Therefore, the list of MSS, and consequently $L'$, can be generated by the Back-and-Forth algorithm.
\end{proof}

\end{document}